\documentclass[]{article}

\usepackage{amsmath, amssymb,amsfonts}

\usepackage[utf8]{inputenc}
\usepackage[english]{babel}
 
\usepackage{amsthm}

\usepackage{hyperref}

\numberwithin{equation}{section}

\newtheorem{thm}{Theorem}[section]

\newtheorem{prop}[thm]{Proposition}

\newtheorem{lemm}[thm]{Lemma}

\newtheorem{cor}[thm]{Corollary}

\newtheorem{defn}[thm]{Definition}

\newcommand{\R}{\mathbb{R}}

\newcommand{\bq}{\begin{equation}}

\newcommand{\eq}{\end{equation}}

\begin{document}

\title{ON THE SOLVABILITY OF EULER GRAPHENE BEAM SUBJECT TO AXIAL COMPRESSIVE LOAD}

\author{Mohamed B. Elgindi
\thanks{ M. B. Elgindi is  with Texas A\&M at Qatar, Doha, Qatar, Email: Mohamad.elgindi@qatar.tamu.edu.}, \ Dongming Wei
\thanks{D. Wei is with the Dept. of Math, Nazarbayev Univesity, Astana, Kazakhstan, Email:dongming.wei@nu.edu.kz.}
and Tarek M. Elgindi}
\maketitle

\begin{abstract} In this paper we formulate the equilibrium equation for a beam made of graphene  subjected to some boundary conditions and acted upon by axial compression and nonlinear lateral constrains as a fourth-order nonlinear boundary value problem.  We first study the nonlinear eigenvalue problem for buckling analysis of the beam. We show the solvability of the eigenvalue problem as an asymptotic expansion in a ratio of the elastoplastic parameters.  We verify that the spectrum is a closed set bounded away from zero and contains a discrete infinite sequence of eigenvalues.   In particular, we prove the existence of a minimal eigenvalue  $\lambda^*$ for the graphene beam corresponding to a Lipschitz continuous eigenfunction, providing a lower bound for the critical buckling load of the graphene beam column. We also proved that  the eigenfunction corresponding to the minimal eigenvalue is positive and symmetric. For a certain range of lateral forces, we demonstrate the solvability of the general equation by using energy methods and a suitable iteration scheme.
\end{abstract}


\bigskip
\section{Introduction}
It is well-known from the materials science, physics, and chemistry literature, that there is intense interest in studying the mechanics of structures made of graphene. Potential industrial applications for graphene made  structures are abundant.  For instance, nanoscale devices that use graphene as basic components, such as resonators, switches, and valves, are being developed in many industries. Understanding the response of individual graphene structure elements to applied loads is therefore crucially important  (see [2], [5], [6], [7] and the references therein for a comprehensive list of applications). In this paper, we analyze the effects of axial compression and nonlinear lateral forces upon an idealized graphene beam. We prove the existence of a minimal "buckling load", which, mathematically speaking, is not obvious due to the structure of the constitutive law relating the stress and strain for a beam made of graphene. Furthermore, we prove the existence and uniqueness of solutions for the equilibrium equation of the elastoplastic beam when the lateral force satisfies a natural bound in terms of the elastoplastic parameters (and we prove non-existence, in certain cases, when this bound is not satisfied).
The Euler buckling of a simply supported straight elastic beam subject to an end axial compressive load   can be modeled by the equation:
\begin{equation} EI v''''+Pv''=0, 0<x<L
\end{equation}
with boundary conditions
\begin{equation} v(0)=v(L)=v''(0)=v''(L)
 \end{equation}
where  $L$  is the length of the beam,   $E$ is Young's modulus, and $I$ is the area moment of inertia.
Integrating (1.1) twice gives
\begin{equation} EIv''+Pv=0 \end{equation}
when the last two boundary conditions from (1.2) are taken into account.
Therefore the boundary value problem (1.1)-(1.2) reduces to the well-known eigenvalue problem for the Laplacian in one dimension:
\begin{equation} EIv''+Pv=0\end{equation}
\begin{equation} v(0)=v(L)=0\end{equation}
As is well known, system (1.4)-(1.5) yields a sequence of eigenvalues  and eigenfunctions, providing the following buckling modes and the corresponding buckling loads $$v_{k}(x)=\text{sin}(\frac{k\pi x}{L}),  P_{k}=EI(\frac{k\pi}{L})^2, k=1,2,3,...$$
Furthermore, each eigenfunction of (1.4)-(1.5) is simple. The first buckling mode $v_{1}$ corresponds to the well-known Euler critical buckling load  $P_{1}=EI(\frac{\pi}{L})^2$, which is sometimes called the onset buckling load. This is used widely in engineering practice for designing weight supporting columns of elastic materials.
\par\
The above Euler critical buckling load is derived based on Hooke's law, relating the axial stress $\sigma_{x}$ and the axial strain $\epsilon_{x}$ by $\sigma_{x}=E\epsilon_{x}$  and the assumption that during the deformation, the cross-sections of the beam column remain perpendicular to the beam's center line.  This classical result is generalized in [12] for Hollomon's law $\sigma_{x}=K|\epsilon_{x}|^{n-1}\epsilon_{x},$ where equation (1.1) is replaced by:
\begin{equation} (KI_n|v''|^{n-1}v'')''+P v''=0 \end{equation}
\begin{equation}  v(0)=v(L)=v''(0)=v''(L) \end{equation}
and the critical load of (1.6)-(1.7)  is given by: $$P_{cr}=\frac{2n(\pi_{2,1+1/n})^2}{n+1}K I_{n}$$
where $I_{n}=\int_{A}|y|^{1+n}dydz$  is the generalized area moment of inertia,  and $\pi_{2,1+1/n}=2\int_{0}^{\frac{\pi}{2}} cos(\theta)^{\frac{n-1}{n+1}}d\theta$ is a generalized Pi.  The first eigenfunction is defined in terms of the generalized sine function associated with $\pi_{2,1+1/n}$ by using the notation of the two parameter function developed in [13].
Graphene material has been shown to be modeled by the following quadratic stress-strain constitutive law (see [3], [4], [6], [8], [9], [10]):
\begin{equation} \sigma_{x}=E\epsilon_{x}+D|\epsilon_{x}|\epsilon_{x}, \end{equation}
where D is related to the Young's modulus by the relation $D=-\frac{E^2}{4\sigma_{\max}}$ called the effective nonlinear (third-order) elastic modulus [3] , and $\sigma_{\max}$  is the material's ultimate maximal shear stress. For small strain, the elastic stress $E\epsilon_x$  dominates (1.8), while the plastic stress $D|\epsilon_x|\epsilon_x$    becomes prominent for large strain. Notice that the ratio $|\frac{D}{E}|=\frac{E}{4\sigma_{\max}}$ is the elastoplastic parameter which we will use in our asymptotic analysis of Section 2.  When this parameter is small then the material's ultimate maximal shear stress $\sigma_{\max}$  is very large, and the elastic behavior dominates.
The equilibrium equation for a grapheme made Euler-Beam subject to axial compressive load $P,$  lateral force $f,$  and a nonlinear support $g$ (all per unit length) is given by the fourth order equation:
\begin{equation} EI v'''' +DI_{2}(|v''|v'')''+Pv'' +g(v'')=f(x), 0<x<L, \end{equation}
where $I=\int_{A} y^2 dy dz,$   $I_{2}=\int_{A} |y|^3 dy dz,$  and the z-axis being the off-plane direction and $A$  is the cross sectional area of the beam. We consider (1.9) along with one of the pin-pin (PP), and pin-slide (PS), or the slide-slide (SS) boundary conditions:
\begin{equation} \text{(PP Conditions)  } v(0)=v(L)=v''(0)=v''(L)=0 \end{equation}
\begin{equation}\text {(PS Conditions)  } v'(0)=v(L)=v'''(0)=v''(L)=0 \end{equation}
\begin{equation} \text{(SS Conditions)  } v'(0)=v'(L)=v'''(0)=v'''(L)=0 \end{equation}
Using the non-dimensional variables and parameters:
$$z=xL^{-1}, u=vL^{-1}, \alpha= \frac{|D|I_{2}}{EIL}, \lambda=\frac{PL^2}{EI}, \hat{g}(u'')=\frac{g(u'')}{EIL^{-3}}, \hat{f}(z)=\frac{f(z)}{EIL^{-3}}, $$
equation (1.9) can be rewritten as:
\begin{equation} u''''-\alpha (|u''|u'')''+\lambda u'' +\hat{g}(u'')=\hat{f}(z).\end{equation}
The boundary conditions (1.10)-(1.12) become:
\begin{equation} \text{(PP Conditions)  } u(0)=u(1)=u''(0)=u''(1)=0 \end{equation}
\begin{equation}\text {(PS Conditions)  } u'(0)=u(1)=u'''(0)=u''(1)=0 \end{equation}
\begin{equation} \text{(SS Conditions)  } u'(0)=u'(1)=u'''(0)=u'''(1)=0\end{equation}
In the next section we will study a special case of (1.13):
\begin{equation} u''''-\alpha(|u''|u'')''+\lambda u'' =0\end{equation}
with the boundary condition
\begin{equation} u(0)=u(1)=u''(0)=u''(1)=0 \end{equation}
Here (1.17)-(1.18) represents the buckling problem for a Euler graphene beam which replaces problems (1.1)-(1.2) and (1.6)-(1.7) for the elastic and Hollomon beams, respectively. In Section 2, we provide an asymptotic expansion of the first eigen-pair of (1.17)-(1.18) in terms of the perturbation parameter $\alpha,$ and prove that, for small enough $\alpha$, each eigen-pair is simple and continuously dependent upon $\alpha$, and we also establish the existence of an infinite sequence of eigen-pairs of (1.17)-(1.18). In section 3, we show that all the eigenvalues of (1.17)-(1.18) are positive and derive a lower bound for the smallest eigenvalue. This lower bound gives us a lower bound to the critical buckling load of the graphene beam
$$P_{cr}\ge \frac{\pi^2EI}{2L^2}$$
which means that the buckling load for the graphene beam column is no less then the Euler-buckling load.
In sections 4 and 5, we consider the global existence and uniqueness of solutions for the boundary value problem equations (1.13) -(1.14), for the case of PP conditions. Similar techniques are valid for the other boundary conditions. This way, we extend the results established in [14] and [15] to the graphene beam with nonlinear support.

\section{ Existence of Eigen-pairs and Buckling Analysis of the Graphene Beam}

We first present some formulations equivalent to the boundary value problem (1.17) and (1.18).
Integrating (1.17) twice, and applying the last two boundary conditions of (1.18) we obtain the nonlinear eigenvalue problem:
\begin{equation} u''-\alpha |u''|
u''+\lambda u=0,  \end{equation}
\begin{equation} u(0)=u(1)=0, \end{equation}
which we call the u-formulation. Another formulation of this problem is obtained upon using the substitution:
$$ v=u''$$
which turns (1.17) and (1.18) into the boundary value problem:
\begin{equation} (v-\alpha |v|
v)''+\lambda v=0,  \end{equation}
\begin{equation} v(0)=v(1)=0, \end{equation}
which we call the v-formulation.
\begin{defn}
If $\lambda,v$ solve (2.3)-(2.4) with $\lambda\in \mathbb{R}$, $v\in H^{1}(0,1)$, with $v(0)=v(1)=0,$ and $v\ne 0$  in the weak sense then, $(\lambda,v)$ is an eigenvalue and eigenfunction of (2.3)-(2.4).
\end{defn} 
Consider an eigen-value problem in the form:
\begin{equation}w''+ \lambda g'(w)=0 \end{equation}
\begin{equation} w(0)=w(1)=0\end{equation}
 The following theorem follows from Proposition 44.35 of Zeidler, 1985, [1]:
\begin{thm}

Suppose that $g:\R\to\R$ is continuously differentiable function, with $g(0)=0$ and $g'(w)w>0$ for all real numbers $w\ne 0$, and there exists constants $c,d>0$ such that following growth condition holds for all $w\in\R$: $|g(w)|\le c(1+|w|^2),|g'(w)|\le d(1+|w|)$. Then problem (2.5),(2.6) has infinitely many eigen-solutions $(\lambda_i,w_i)$, with $w_i\ne 0,\lambda_i >0$ for all $i\in Z^{+}$ such that
$w_i\rightharpoonup 0$ in $H^1_0(0,1)$ as well as $\lambda_i\to \infty$ as $i\to \infty$.
\end{thm}
\begin{proof} To prove existence of eigen-solutions to problem (2.3) and (2.4), we show that it can be transformed into an equivalent quasi-linear problem like (2.5)-(2.6). For this purpose, we define the derivative of a continuously differentiable even function $g$ by letting it's derivative to be
 \begin{equation}g'(w)=\begin{cases} q^{-1}(w),|w|\le \frac{1}{4\alpha}\\
\frac{1}{2\alpha}sign (w), |w|>\frac{1}{4\alpha} \end{cases}\end{equation}
where $q(v)=v-\alpha|v|v,|v|\le \frac{1}{2\alpha}$, $w=q(v)$.
\end{proof}
\begin{prop} Let $g'$ be the function defined by (2.7), then problem (2.5),(2.6) has infinitely many eigen-solution $(\lambda_i, w_i)$, with $w_i\ne 0,\lambda_i >0$ for all $i\in Z^{+}$ such that
$w_i\rightharpoonup 0$  in $H^1_0(0,1)$ as well as $\lambda_i\to \infty$ as $i\to \infty$.
\end{prop}
\begin{proof}
By definition, $g':\R\to\R$ is odd, continuous, and it satisfies $g'(0)=0$, $g'(w)w>0$ for all real numbers $w\ne 0$. Let $g(w)=\int_0^wg'(z)dz$, then $g(0)=0$ and $g$ is an even function. By explicitly solving $g$, one can verify that the growth conditions are also satisfied. Therefore, the Theorem follows from Theorem 2.2.
\end{proof}
\begin{thm}
Problem (2.3),(2.4) has infinitely many eigen-solution $(\lambda_i, v_i)$, with $v_i\ne 0,$  $v_i \in H^{1}_0(0,1)$, and $\lambda_i >0$, for all $i\in Z^{+}$.
\end{thm}
\begin{proof}
This follows from   $v_i=q^{-1}(w_i)$ and the above Proposition.
\end{proof}

\section{ A Lower Bound of the Eigenvalues and Some Properties of the Eigenfunctions}

In this section we verify some properties of the eigenvalues, prove the existence of a minimumal (first) buckling load, and derive a lower bound of the smallest eigenvalue.
Physically, it makes sense that there be a minimal positive eigenvalue--a so-called critical buckling load. We also give a-priori estimates on the eigenfunctions in the $W^\infty$ norm.

One technicality which gives us a little bit of trouble is that an eigenfunction $v$ is not necessarily smooth in $(0,1).$ In fact, if it has interior zeros it cannot be smoother than of class $W^{2,\infty}(0,1)$ near those zeros (due to the presence of the $|v|$ in our equation). Nonetheless, away from the zeros of a continuous eigenfunction, and away from points where $|v|=\frac{1}{2\alpha},$ the eigenfunction must be smooth.
This can be proved using the same techniques as are used in the regularity part of the proof of theorem 4.1 of section 4.
Note that Theorem 2.2 gives us the existence of $H^2$ eigenfunctions for (3.1)-(3.2).  We first show that there exists a lower bound for all the eigenvalues:
 \begin{thm}
If $\lambda$ is an eigenvalue of (2.3)-(2.4) then $\lambda\geq \frac{\pi^2}{2}.$
\end{thm}
\begin{proof} Proof. Let $w$ be an eigenfunction associated to $\lambda$. Let $v=g'(w)$, then
by definition, $g'(w)=q^{-1}(w)$, therefore  $w=q(v)$.   Solving the quadratic equation, we get $v=\frac{1-\sqrt{1-4\alpha w}}{2\alpha}=\frac{2w}{1+\sqrt{1-4\alpha w}}$ for $0\le v\le\frac{1}{2\alpha}$ and $v=\frac{-1+\sqrt{1+4\alpha w}}{2\alpha}=\frac{2w}{1+\sqrt{1+4\alpha w }}$ for $ -\frac{1}{2\alpha}\leq v$. In either case, we have $|g'(w)|=|v|\le 2|w|$. Since
$$\int_0^1|w'|^2dx=\lambda\int_0^1g'(w)wdx, $$
We have
$$\int_0^1|w'|^2dx \le 2\lambda\int_0^1w^2dx$$
which implies $\pi^2 \le \frac{\int_0^1|w'|^2dx }{\int_0^1w^2dx}\le 2\lambda$ by the Poincar\'{e} inequality.  Therefore,  $\frac{\pi^2}{2}\le \lambda$.
The theorem is proved.
\end{proof}
A-priori, we don't know that there exists a smallest eigenvalue from the previous result. In the following theorem, we will prove the existence of a minimal eigenvalue for (2.3)-(2.4).
\begin{thm}
There exists a minimal eigenvalue $\lambda^*$ of (2.3)-(2.4) satisfying $\lambda^*\ge \frac{\pi^2}{2}$.
\end{thm}
Indeed, take a (decreasing) minimizing sequence of eigenvalues $\lambda_{k}\rightarrow \lambda^*$ where $\lambda^*= \inf \{\lambda: \text{$\lambda$ is an eigenvalue of (3.1)-(3.2)}\}.$ We will show that $\lambda^*$ is an eigenvalue of $(3.1)-(3.2).$  Now, by Theorem 3.2, the associated eigenfunctions, satisfy:$|v_k| \leq \frac{1}{2\alpha}, \,\,\, |v_k'|\leq \frac{\lambda_k}{2\alpha}.$ Since $\lambda_k$ are uniformly bounded, we see that $v_k$ are uniformly bounded in $W^{1,\infty}$ this implies that (upon passing to a subsequence) $v_k$ converges strongly to some $v$ in $W^{1-\epsilon,2}$ for all $\epsilon>0.$ By the Sobolev imbedding theorems and trace theorem we see that $v\in H^{1}_0(0,1)$ and $v_k\rightarrow v$ uniformly. This implies that $|v_k|v_k\rightarrow |v|v$ uniformly. Hence, $v\in H^1_0$ is a weak solution of: $(v-\alpha v|v|)''=-\lambda^* v$ and by the previous theorem, we have$\lambda^*\ge \pi^2$.

Since $\lambda=\frac{PL^2}{EI}$, this Theorem gives us a lower bound on smallest  the buckling load $P_{cr}$ for the graphene beam as $\frac{\pi^2 EI}{L}$, which is the Euler-buckling load for the corresponding linear beam model.
In the following, we show some regularity and properties of the eigenfunctions.
\begin{lemm}
If $v(x)\not=\frac{1}{2\alpha}$ then $v$ is smooth in a neighborhood of $x.$
\end{lemm}
This lemma is a consequence of standard elliptic regularity theory. One can mimic the arguments in Section 4, for example.
\begin{thm}
Let $v$ be a continuous eigenfunction of (2.3)-(2.4) corresponding to an eigenvalue $\lambda$. Then $|v|\le\frac{1}{2\alpha}$ on $[0,1]$.
\end{thm}
 \begin{proof}Suppose that the conclusion of the theorem is false. Then without loss of generality we may assume that $v$ has a maximum at $x=c$ and that $v(c)>\frac{1}{2\alpha}$.  Now, because $v$ has a local maximum at $c,$ $v''(c)\le 0$ and $v'(c)=0$.  Expanding equation (3.1) gives:
\begin{equation}
v''(1-2\alpha|v|)-2\alpha\frac{vv'^2}{|v|}+\lambda v=0.
\end{equation}
Letting $x=c$ and noting that $\lambda>0$ gives a contradiction.
\end{proof}
\begin{lemm}
Suppose that $v(a)=0$ for some $a\in [0,1]$ and suppose that $b$ is the first number larger than $a$ for which $v(b)=0$. Then, $$|v'|\leq \frac{\lambda}{2\alpha}.$$
\end{lemm}
Without loss of generality, we may assume that $v> 0$ in $(a,b)$. In this case, $$v''(1-2\alpha v)-2\alpha v'^2+\lambda v=0$$ except, perhaps, at the points where $|v|=\frac{1}{2\alpha}.$ Now, it must be that $v'(a)>0$ since otherwise $v'(a)=0$ in which case the standard uniqueness theorem for ODE's will imply that $v\equiv 0$. Similarly $v'(b)<0$. Thus, $v''(a)>0$ and $v''(b)>0$. Hence, $v'$ must achieve its maximum inside $(a,b).$ Say the maximum is achieved at $c$. Now, $v(c)\not=\frac{1}{2\alpha}$ because otherwise $v'(c)=0$ and $v'$ couldn't achieve its max at $c$. Hence $v$ is infinitely differentiable at $c$. Then, $v''(c)=0$ and $$v'(c)^2=\frac{\lambda}{2\alpha}v(c)\leq\frac{\lambda^2}{(2\alpha)^2} \implies v'(c)\leq \frac{\lambda}{2\alpha}. $$ Now there are two cases: $v'$ achieves its minimum in $(a,b)$ or $v'$ achieves its minimum at $b.$ If the minimum is achieved in the interior then it must be bounded from below by $\frac{-\lambda}{2\alpha}$. In this case $$|v'|\leq \frac{\lambda}{2\alpha}.$$ On the other hand, by integrating (2.3) from $a$ to $b$ we see (using that $v(a)=v(b)=0$) that: $v'(b)-v'(a)=-\lambda\int_{a}^b v(x)dx.$
This implies that $v'(b)=v'(a)-\lambda\int_{a}^{b}v(x)dx\geq -\frac{\lambda}{2\alpha}.$
Thus, in all cases:$|v'|\leq \frac{\lambda}{2\alpha}$
\begin{cor} Let $(\lambda,v)$ be an eigen-solution to (2.3)-(2.4), then
$$|v'(x)|\leq \frac{\lambda}{2\alpha},  \forall x \in [0,1].$$ Consequently,
$$|v''(x)|\leq \frac{\lambda(\lambda+1)}{2\alpha\big|1-2\alpha|v(x)|\big |},\forall x \in [0,1]$$
\end{cor}
It is well-known that for the corresponding problem with $\alpha=0$, the eigenfunction is unique up to a sign difference. However, we are able to show that for the case  $\alpha\ne 0$, the following theorem
We now show that
\begin{thm} Let $w_1$ and $w_2$ be two eigenfunctions associated with
The  eigen-function associated with the smallest eigenvalue $\lambda^*$ of (3.1)-(3.2). If $||w_1-w_2||_{1,\infty}\le \frac{1+1/\pi}{2\alpha}$, then  $w_1=w_2$.
\end{thm}
This is a special case of proposition 5.1 below.

\begin{thm} Let $\lambda^*$ be the smallest eigenvalue and let $$S^*=\{w:-w''=\lambda^*g'(w).w(0)=w(1)=0\}.$$ Then, for every $w\in S$ $$w(a)>0,\,\,\, \forall \,a\in(0,1).$$Furthermore, every $w\in S^*$, satisfies $w'(1/2)=0$ and $w(x)$ is symmetric about $x=1/2$.
\end{thm}

\begin{proof} If $w^*(a)=0$ for some $w^*\in S^*$, and $0<a<1$.  Let $w_a(x)=w^*(ax),x\in[0,1]$, then $w_a(0)=w_a(1)=0$, and $-w_a''(x)=-a^2w*''(ax)=-a^2\lambda^*g'(w^*(ax))=
-a^2\lambda^*g'(w_a(x))$. Therefore, $w_a$ is an eigenfunction with eigenvalue $\lambda=a^2\lambda^*$ which is a contradiction to the fact $\lambda^*$ is the smallest among all members of $S^*$ satisfying $w(a)\ne 0$ for $0<a<1$.
Let $v(x)=w(1-x)$, then $v'(x)=-w'(1-x)$ gives $v'(0)=-w'(1)$.  Multiply both sides of $-w''=\lambda g'(w)$ by $w'$, the integrate both sides using $g'(0)=0,w(0)=0$, we get $w'^2(x)=2\lambda g(w(x))+(w'^2(0)-\frac{\lambda}{6\alpha})$. Setting $x=1$, we have $w'^2(1)=w'^2(0)$ Taking $w(x)\in\S^*,w'(0)=1$, then
$w'(1)=-1$ and $v'(0)=-w'(1)=1$. We have shown that $v$ a solution to the initial value problem $v''=-\lambda g'(v),v(0)=0,v'(1)=-1$. Since $g'(w)$ is Lipschitz continuous, the initial value problem has a unique solution, we conclude that $w(x)=v(x)=w(1-x), \forall x \in[0,1]$.
\end{proof}\break
\noindent For each $w\in S^*$, define $$w_2(x)=\begin{cases}w(2x), &0\le x\le 1/2,\\-w(2(1-x)), &1/2\le x\le 1.\end{cases}$$ \break
\noindent Let $\lambda_2=4\lambda^*$, then $w_2$ is an eigenfunction with eigenvalue $\lambda_2$. This procedure can be used to generate a sequence of eigenvalue
$\lambda_k=k^2\lambda^*$ with eigenfunction $w_k$ defined similarly for $k=2,3,...$. 
\begin{cor} Let $\lambda^*$ be the smallest eigenvalues of (2.3)-(2.4), then
$\{\lambda_k=k^2\lambda^*, k=1,2,...\}$ is a sequence of eigenvalues.
\end{cor}

\section{ Approximation of the Eigen-values}
When $\alpha=0$, (2.1)-(2.2) reduces to the eigenvalue problem for the Euler elastic beam:
\begin{equation}u''+\lambda u=0 \end{equation}
\begin{equation} u(0)=u(1)=0\end{equation}
whose eigenpairs are given by:
\begin{equation} \lambda_{k}=(k\pi)^2, u_{k}=sin(\pi k z), k=1,2,3,...\end{equation}
In particular this linear problem has a discrete spectrum and each eigen-value is simple.
Consider the nonlinear graphene operator defined by:
$$N_{G}(\alpha,u)=u''-\alpha|u''|u'',u\in H^{2}(0,1)\cap H^{1}_{0}(0,1)$$
Ideally, we would like to prove that $N_{G}$ has a discrete spectrum. The next proposition is a first step in this direction. We show that for each eigenvalue $\lambda_{k}$ of the linear operator (the Laplacian) there exists a continuously differentiable curve of eigenvalues to $N_{G}(\alpha,\cdot),$ for small $\alpha.$
The proof of these facts is based on the implicit function theorem as demonstrated below.

Note that we will use the notation $<f,g>$ to denote the $L^2$ inner product of $f$ and $g$:

$$<f,g>= \int_{0}^{1} f(x)g(x)dx$$

\begin{prop} For each eigenpair $(u_{1},\lambda_{1})$ of (2.5)-(2.6), there exists $\alpha_0>0$ so that there exists a unique smooth curve $(u(\alpha),\lambda(\alpha))$ of eigenpairs of $N_{G}(\alpha,u)$ defined for $\alpha \leq \alpha_{0}$ such that $\lambda(0)=\lambda_{1}$ and $u(0)=u_{1}.$
\end{prop}
Define $F:\mathbb{R}\times H^2 \cap H^{1}_{0}\times \mathbb{R} \rightarrow L^{1}\times \mathbb{R}$ in the following way:
$$F(\alpha,u,\lambda)=\big( u''-\alpha|u''|u''+ \lambda u,<u_{1}',u'>-<u_{1}',u_{1}'>\big)$$

$F$ is continuously differentiable and $F(0,u_{1},\lambda_{1})=(0,0).$
Now, we seek to prove that $F_{u,\lambda}(0,u_{1},\lambda_{1})=  \left[ {\begin{array}{lcl}{}
             ( \cdot)''+\lambda_{1}(\cdot) & u_{1} \\
             < u'_{1}, (\cdot)'> & 0 \\
                \end{array} } \right ]$ is invertible. It is clealy surjective.

Now, if $F_{u,\lambda}(0,u_{1},\lambda_{1}) \left[ {\begin{array}{lcl}{}
             u  \\
             \lambda \\
                \end{array} } \right]=0$ then $u$ and $\lambda$ have to satisfy the following system:
$$u''+\lambda_{1}u+\lambda u_{1}=0$$ and
$$<u_{1}',u'>=0.$$

Multiplying the first equation by $u_{1},$ integrating from $0$ to $1,$ and integrating by parts in the first term, and using the fact that $u''_1 +\lambda_{1}u_1 =0,$ we get:$\lambda \int_{0}^{1} u_{1}^2 dx=0$ so that $\lambda=0.$
Then the first equation becomes
$$u''+\lambda_{1}u=0.$$ But since $<u'_{1},
u'>=0$ and since $u_{1}$ is simple, $u\equiv 0.$ The proposition then follows from the implicit function theorem in Banach spaces.
We now seek to find an asymptotic expansion of the solution of (2.1)-(2.2) in powers of $\alpha$ .
The zeroth order boundary value problem is (2.5)-(2.6) whose solution is given by (2.7).
The first order equation then reads:
$$u''_{2}+\lambda_{1}u_{2}=|u''_{1}|u''_1-\lambda_{2}u_{1} $$
$$u_{2}(0)=u_{2}(1)=0 $$
whose solvability condition gives: $$\lambda_{2}=-\frac{\int_{0}^{1}|u''_{1}|u''_1u_1dz}{\int_{0}^{1}|u_{1}|^2dz}$$
In this way we obtain an asymptotic expansion:
$$u(z)=u_{1}(z)+\alpha u_{2}(z)+O(\alpha^2)$$
$$\lambda=\lambda_{1}+\alpha \lambda_{2} +O(\alpha^2)$$
valid for small enough $\alpha$, where $u_{2}$ is the unique solution of the first order problem above.

\section{Existence and Uniqueness of the Beam with Nonlinear Support}

In this section we want to prove existence and uniqueness of solutions for the elastic beam equations with compression below the first buckling load, with a nonlinear foundational support, and subject to a mild external force. In the next section we will show that the conditions we assume to prove existence and uniqueness are more or less optimal.
Consider the following non-linear elliptic boundary value problem:
\begin{equation} ((1-2\alpha|v|) v')'+\lambda v+ g(v)=f, \, \, \text{in} \, \, (0,1) \end{equation}
\begin{equation} v(0)=v(1)=0\end{equation}
with $\alpha\geq 0,$ $\lambda < \frac{\pi^2}{2},$ and $f$ is a bounded function. Furthermore, $g$ is a differentiable function which is homogeneous of degree 2 or more and satisfies the following inequalities:
$$tg(t)\leq 0, g'(t)\leq 0 \, \, \text{for all}\, \, t.$$
The main result of this section is that if $f$  is small enough in $L^{2}(0,1)$, then (4.1)-(4.2) has a unique $H^2$ solution. Moreover, we show by example, that our result is in some sense optimal: if $f$ is positive and large enough then no solution exists.
We prove the uniqueness before we prove existence.
We prove that if we have two solutions of (4.1)-(4.2) which are both small enough then the two solutions must coincide.
Define the following classes of functions:
$$B_{\delta} \equiv \{k\in W^{1,\infty}(0,1): \, |k|_{W^{1,\infty}}\leq \delta \}$$
\begin{prop}
If $\delta<\frac{(1+\frac{1}{\pi})}{2\alpha},$ then (4.1)-(4.2) has at most one weak solution in $B_{\delta}.$
\end{prop}
Suppose that $v_{1},v_{2} \in B_{\delta}$ solve (1.1)-(1.2).
Then, $v=v_1 - v_2$ satisfies the following equation:
$$v'' + \lambda v +g(v_1)-g(v_2)= 2\alpha[((|v_1|-|v_2|)| v'_1)'- (|v_1|v')'] .$$
Now multiply by $v$ and integrate by parts. Since $v=0$ at 0 and 1, all the boundary terms vanish and we get:
\begin{equation}\begin{array}{lcl} \int_{0}^1 | v'|^2 +|\lambda|\int_{0}^1 |v|^2 -\int_{0}^{1} (g(v_{1})-g(v_{2}))(v_1 -v_2) \\
\leq 2\alpha [ \int_{0}^1 |v||  v'_1||v'|+ |v_1||v'|^2] \end{array} \end{equation}
where we used $$||v_1|-|v_2||\leq |v_1-v_2|$$
Because $g'\leq 0,$ we have that $(g(v_{1})-g(v_{2}))(v_1 -v_2) \leq 0,$ so we can drop the last term on the left hand side of (4.3).
By the Poincar\'{e} inequality, we have  $|v|_{L^2} \leq \frac{1}{\pi}|v'|_{L^2}.$ This and using the fact that $|v_1|,|v'_1| \leq \delta,$ we see that
$$| v'|_{L^2} \leq 2\alpha \delta(1+\frac{1}{\pi})| v'|_{L^2}$$
Therefore, if $\delta <(1+\frac{1}{\pi})^{-1} \frac{1}{2\alpha},$  then $v'\equiv 0$ and, using the boundary condition, the uniqueness theorem is proven.
The proof of existence will rely upon energy estimates and a suitable iteration scheme.
We will begin by proving the existence of a small solution in $H^{1}_{0}$ under a suitable condition on $f.$ 
In (4.1)-(4.2), we write $v=\frac{1}{2\alpha} w$ and $F=2\alpha f.$ Then we get that $v$ is a solution of (4.1)-(4.2) if and only if $w$ is a solution of:
\begin{equation}  ((1-|w|) w')'+ \lambda w + 2\alpha g(\frac{1}{2\alpha} w)=F, \, \, \text{in} \, \, (0,1)
\end{equation}
\begin{equation} w(0)=w(1)=0. \end{equation}
Recall that $\lambda < \frac{\pi^2}{2}$ and $G(\cdot):= 2\alpha g(\frac{1}{2\alpha} \cdot)$ satisfies the same conditions as $g.$

The main idea we want to use is that if $F$ is smooth and small enough, then, using the maximum principle, $w$ must also be small. Once $w$ is small, the equation becomes uniformly elliptic and we will then be able to deduce the existence and uniqueness of a small solution. We now prove existence of an $H^{1}$ weak solution.

\begin{prop}

Let $h$ be a bounded, measurable function with $|h|\leq \frac{1}{2}.$ Then if $w$ solves the following semi-linear boundary-value problem
\begin{equation} ((1-|h|)w')'+{\lambda} w +G(w)=F, \, \, \text{in} \, \, (0,1) \end{equation}
\begin{equation} w(0)=w(1)=0, \end{equation}
with $tG(t)\leq 0,$ for all $t$. Assume further that $\lambda < \frac{\pi^2}{2}.$ Then $$|w'|_{L^2} \leq \frac{1}{\pi(\frac{1}{2} - \frac{\lambda}{ \pi^2} )}|F|_{L^2}$$
\end{prop}
Multiply (4.6) by $w$ and integrate from $0$ to $1.$ Upon integrating by parts we see
\begin{equation} \int_{0}^{1}(1-|h|)|w'|^2 dz-\lambda\int_{0}^{1} w^2 dz -\int_{0}^{1} G(w)wdz = -\int_{0}^{1}Fw dz \end{equation}
Using the condition on $h$ and that $tG(t)\leq 0,$ we see that $$\frac{1}{2} \int_{0}^{1} |w'|^2 dx \leq|\int_{0}^{1}Fw dx|  + \lambda\int_{0}^{1}|w|^2 dx. $$
Now, using the best constant in the Poincar\'{e} inequality on $[0,1],$ we know that $$\int_{0}^{1} |w|^{2} dz \leq \frac{1}{\pi^2} \int_{0}^{1} |w'|^2 dz$$
This implies that $$\frac{1}{2} \int_{0}^{1} |w'|^2 dz \leq|\int_{0}^{1}Fw dz|  + \frac{{\lambda}}{\pi^2}\int_{0}^{1}|w'|^2 dz $$
Since, by assumption, ${\lambda} < \frac{\pi^2}{2},$
$$(\frac{1}{2} - \frac{\lambda}{ \pi^2} )\int_{0}^{1} | w'|^2 dz \leq  |\int_{0}^{1}Fu dz| $$
Using the Cauchy-Schwarz inequality and the Poincar\'{e} inequality once more we see:
$$|w'|_{L^2} \leq \frac{1}{\pi(\frac{1}{2} - \frac{\lambda}{ \pi^2} )}|F|_{L^2} $$
We now need the fact that $H^{1}_{0}$ is imbedded in $L^\infty$ in dimension one:
$$|f|_{L^\infty} \leq |f'|_{L^2} $$ for all $f\in H^{1}_{0}(0,1).$ This is just a consequence of the Cauchy-Schwarz inequality. Therefore,
$$|w|_{L^\infty} \leq \frac{1}{\pi(\frac{1}{2} - \frac{\lambda}{\pi^2} )}|F|_{L^2}$$
One may also try to prove this proposition using the maximum principle by seeing that the condition on $h$ implies that the ellipticity constant of our equation is $1-|h|\geq \frac{1}{2}.$ However we wanted a simple way to get exact constants in our bounds. Now assume $|F|_{L^2}\leq \frac{\pi(\frac{1}{2} - \frac{\lambda}{ \pi^2})}{2}$ and  define the following sequence of functions $w_{n}:$
$$w_0 =0, ((1-|w_{n-1}|)w'_n)'+w_n+G(u_{n})=F, \, \, \text{in} \, \, (0,1) $$
$$ w_n (0)=w_n (1)=0 $$
Using the theory of semi-linear elliptic equations in one dimension, we see that the sequence $w_{n}$ can be defined for all $n$ (see, for example, [9]). Moreover, by Proposition 5.2, $|w_n|\leq \frac{1}{\pi(\frac{1}{2} - \frac{\lambda}{ \pi^2} )} |F|_{L^2}$ and  $|w_n|_{H^1} \leq  \frac{2}{\pi(\frac{1}{2} - \frac{\lambda}{ \pi^2} )}|F|_{L^2}$ for all $n.$
Thus the sequence $w_n$ is uniformly bounded in $H^1 \cap L^\infty.$ Thus we may extract a subsequence of $w_n$ which converges weakly in $H^1 ,$ strongly in $L^p,$ for some $p>2$ and pointwise to a function $w \in H^1 \cap L^\infty.$ Moreover, $|w|_{H^1} \leq \frac{2}{\pi(\frac{1}{2} - \frac{\lambda}{ \pi^2} )}|F|_{L^2}$ and $|w|_{L^\infty} \leq \frac{1}{\pi(\frac{1}{2} - \frac{\lambda}{ \pi^2} )} |F|_{L^\infty}.$
Therefore, $w$ is a bounded weak solution of (4.6)-(4.7).
We now want to show that $w,$ in fact, belongs to $H^2$ with a certain smallness estimate. We aim to show that $w\in H^2(0,1)$ with an appropriate bound. We will first show that $w-\frac{|w|w}{2} \in H^2(0,1).$ Notice that $(w-\frac{|w|w}{2})''=((1-|w|)w')'.$ Therefore, we can write or equation as
$$(w-\frac{|w|w}{2})''=H$$ where $H$ is an $L^2$ function  ($H=F-{\lambda}w-G(w)$). Using standard elliptic theory,  $w-\frac{|w|w}{2} \in H^{2}(0,1).$
Let $v=w-\frac{w|w|}{2}.$ Define the function $\Phi$ with $\Phi(x)=x-\frac{1}{2}x|x|.$ Now, $\Phi$ is not invertible on the whole real line. However, it is invertible on $|x|\leq \frac{1}{2}.$ From Theorem 3.2,  we have  $|w|\leq \frac{1}{4},$ $\Psi$ is well-defined. Noting that $\Phi'(x)=1-|x|,$ we see that $\Phi$ is invertible for $|x|\leq \frac{1}{2}.$ Call the inverse $\Psi.$ By the inverse function theorem, $|\Psi'|\leq 2.$ In fact, $\Psi \in W^{2,\infty}$. So, $\Psi(v)=w.$  Now we want to transfer our regularity estimate for $v$ to a regularity estimate for $w.$ This follows by the chain rule in Sobolev spaces.
Now that $w\in H^2,$ we can perform the following estimates: Take the equation $$((1-|w|)w')'+{\lambda}w+G(w) = F$$ and multiply by $((1-|w|)w')'$ then integrate from 0 to 1.
Recall that $G'(t)\leq 0.$ Then we see, upon integration by parts in the second and third terms,
$$\begin{array}{lcl}\,\int_{0}^{1} |((1-|w|) w')'|^2 dz -\lambda\int_{0}^{1}|w'|^2(1-|w|)dz+\int_{0}^{1} G'(w)|w'|^2(1-|w|)dz\\=\int_{0}^{1} F((1-|w|)w')'dz\end{array}  $$
Therefore, $$\begin{array}{lcl}\,\int_{0}^{1} |(w-\frac{|w|w}{2})''|^2 dz -{\lambda}\int_{0}^{1}|w'|^2(1-|w|)dz-\int_{0}^{1} G'(w)|w'|^2(1-|w|)dz\\=\int_{0}^{1} F(w-\frac{|w|w}{2})''dz \end{array} $$
Using the Cauchy-Schwarz inequality,
$$\int_{0}^{1} |(w-\frac{|w|w}{2})''|^2 dz -{\lambda}\int_{0}^{1}|u'|^2dz\leq \frac{1}{2} \int_{0}^{1} F^2+|(w-\frac{|w|w}{2})''|^2dz  $$
$$ \int_{0}^{1} |(w-\frac{|w|w}{2})''|^2 dz \leq  \int_{0}^{1} F^2 dz+2\lambda\int_{0}^{1}|w'|^2dz  $$
Now recall that ${\lambda}\leq \frac{\pi^2}{2}$ and $|w'|_{L^2}^2 \leq(\frac{1}{\pi(\frac{1}{2} - \frac{\lambda}{ \pi^2} )})^2 |F|_{L^2}^2 $.
Therefore, $$ \int_{0}^{1} |(w-\frac{|w|w}{2})''|^2 dz \leq (1+(\frac{1}{(\frac{1}{2} - \frac{\lambda}{ \pi^2} )})^2) \int_{0}^{1} F^2 dz  $$
So, $$ \int_{0}^{1} |v''|^2 dz \leq (1+(\frac{1}{(\frac{1}{2} - \frac{\lambda}{ \pi^2} )})^2) \int_{0}^{1} F^2 dz $$ and $w=\Psi(v).$ By simple calculations, $|w''|_{L^2} \leq 2(|v''|_{L^2}+|v_{z}^2|_{L^2}).$
Therefore, $$ |w|_{H^2} \leq 2 (\sqrt{(1+(\frac{1}{(\frac{1}{2} - \frac{\lambda}{ \pi^2} )})^2})|F|_{L^2}+(1+(\frac{1}{(\frac{1}{2} - \frac{\lambda}{ \pi^2} )})^2)|F|_{L^2}^2)$$
Thus we have proven the following theorem:
\begin{thm}
Let $\alpha>0$ be given. Suppose that $g$ is a continuous function on the real line which is homogeneous of degree 2 or more. Suppose further that $tg(t)\leq 0$ and $g'(t)\leq 0$ for all $t.$ Suppose ${\lambda}\leq \frac{\pi^2}{2}.$ Then there exists $c_{1}>0$ small (explicitly given below) so that if $f$ is a measurable $L^2$ function on $[0,1]$ with $|f|_{L^2}\leq \frac{c_{1}}{\alpha},$ then the following non-linear boundary-value problem has a unique solution belonging to $H^2.$
\begin{equation} ((1 -2\alpha|v|) v')' +\lambda v+g(v)=f \end{equation}
\begin{equation} v(0)=v(1)=0. \end{equation}
Moreover, there exists a constant $c_{2}$ so that $|v|_{H^2}\leq \frac{c_{1}}{c_{2}}|f|_{L^2}.$ Here,
$c_{1}= \frac{\pi(\frac{1}{2} - \frac{\lambda}{\pi^2})}{2}, $ and
$c_{2}=4\sqrt{(1+(\frac{1}{(\frac{1}{2} - \frac{\lambda}{\pi^2} )})^2}).$
\end{thm}

\section{Nonexistence for Large External Force}

\begin{prop} Consider the system (4.9)-(4.10). Take $\lambda=0$ and $g=0.$ Then there exists a universal constant $c_{3}>0$ so that if we take $f\equiv \frac{c}{\alpha},$ for $c> c_{3},$ then there exists no solution to $(4.9)-(4.10).$
\end{prop}

Problem (4.9)-(4.10) reduces to $$ (1-\alpha|v|v)''=\frac{c}{\alpha},$$ $$v(0)=v(1)=0.$$
Integrating twice and using the boundary condition yields $$ v-\alpha|v|v=\frac{c}{\alpha} z(z-1)$$
Factoring we get:
$$v(1-\alpha|v|)=\frac{c}{\alpha}z(z-1).$$
Since the right hand side is never zero in (0,1), the left hand side can never be zero either. Therefore, $v$ is either positive or negative in $(0,1).$
Moreover, since $v(0)=0,$  $(1-2\alpha|v|)>0$ for $z$ close to 0. The right hand side is negative in $(0,1).$ Therefore $v$ is negative for $z>0$ small. Therefore $v$ is negative in the entire interval $(0,1)$. Therefore, $$ v+\alpha v^2=\frac{c}{\alpha} z(z-1)$$ Take $z=\frac{1}{2}.$ Then,$$ v(\frac{1}{2})+\alpha v^2 (\frac{1}{2})=-\frac{1}{4}\frac{c}{\alpha}$$
So, if $c>1,$ we see that the discriminant of this equation is negative so that no solutions exist.
Note that in the case that $\lambda=0,$ $c_{1}=\frac{\pi}{4}<2$ so that if the external force is less than $\frac{\pi}{8\alpha}$ in $L^2,$ then we have existence and uniqueness of an $H^2$ solution. Moreover, we have an example of an external force larger than $\frac{1}{\alpha}$ in $L^2$ so that there exists no $H^1$ solution to (4.9)-(4.10). Thus our result in theorem 3.3 is, essentially, optimal both in the mathematical and physical sense. Physically, this says that for a small enough lateral force we have a smooth deformation, but for a large lateral force--'small' and 'large' being determined by the basic physical constants in the system such as the maximal strain $\sigma_{\max},$ there is no smooth deformation.
\section{Acknowledgements}
This research was done while D. Wei  was visiting Texas A\&M University in Qatar in summer 2013 and he acknowledges the gracious support of TAMUQ and the Qatar Foundation. T. Elgindi was supported by NSF grant no. 1211806 during the completion of this research.

\end{document}